\documentclass[a4paper,11pt]{article}

\usepackage{microtype}
\usepackage[T1]{fontenc}
\usepackage[utf8]{inputenc}
\usepackage{hyperref}

\usepackage{tgtermes}
\usepackage{amsfonts}
\usepackage{amsmath}
\usepackage{amssymb}
\usepackage{amsthm}
\usepackage{thm-restate}
\usepackage{nicefrac}
\usepackage{mathtools}  
\usepackage{enumitem}

\DeclareMathOperator*{\argmin}{arg\,min}  
\newcommand{\pdim}{\operatorname{P{\mkern-1mu}dim}}
\newcommand{\cost}{\operatorname{cost}}
\newcommand{\supp}{\operatorname{supp}}
\newcommand{\poly}{\operatorname{poly}}
\newcommand{\val}[1]{\operatorname{val}(#1)}

\newtheorem{theorem}{Theorem}

\newtheorem{lemma}[theorem]{Lemma}
\theoremstyle{definition}
\newtheorem{definition}[theorem]{Definition}

\title{Learning-Augmented Maximum Flow}

\author{Adam Polak\thanks{Supported by the Swiss National Science Foundation project \emph{Lattice Algorithms and Integer Programming} (185030).} \\ {\normalsize EPFL} \and Maksym Zub \\ {\normalsize Jagiellonian University}}

\date{}

\begin{document}

\maketitle

\begin{abstract}
We propose a framework for speeding up maximum flow computation by using predictions. A prediction is a flow, i.e., an assignment of non-negative flow values to edges, which satisfies the flow conservation property, but does not necessarily respect the edge capacities of the actual instance (since these were unknown at the time of learning). We present an algorithm that, given an $m$-edge flow network and a~predicted flow, computes a maximum flow in $O(m\eta)$ time, where $\eta$ is the $\ell_1$ error of the prediction, i.e., the sum over the edges of the absolute difference between the predicted and optimal flow values. Moreover, we prove that, given an oracle access to a~distribution over flow networks, it is possible to efficiently PAC-learn a~prediction minimizing the expected $\ell_1$ error over that distribution. Our results fit into the recent line of research on learning-augmented algorithms, which aims to improve over worst-case bounds of classical algorithms by using predictions, e.g., machine-learned from previous similar instances. So far, the main focus in this area was on improving competitive ratios for online problems. Following Dinitz et al.~(NeurIPS 2021), our results are one of the firsts to improve the running time of an offline problem.
\end{abstract}

\section{Introduction}

Computing a maximum $s$-$t$ flow in a flow network (i.e., in a directed graph with nonnegative edge capacities and designated source and sink nodes) is a basic problem in combinatorial optimization. It is a building block of a number of more advanced algorithms, with relevance both in theory (e.g., in graph algorithms and scheduling) and in practice (e.g., in computer vision and transport).

Imagine we are to solve multiple similar instances of the maximum flow problem, e.g., the instances are drawn at random from a distribution, or they are snapshots of a single underlying instance changing over time. Can we learn an approximate shape of optimal solutions, and then use it to speed up further computations? Or, to put it differently, assume we have a solution -- e.g., obtained from past data or computed by a very fast heuristic -- that is not necessarily optimal, maybe not even feasible, but close to an optimal solution. How can we use such an imperfect solution to warm-start a maximum flow algorithm and get a better running time?

Warm-starting maximum flow algorithms have been studied in the past heuristically (e.g., in computer vision, where maximum flow is often used to compute minimum cuts in subsequent frames of a video~\cite{JuanB06}). In contrast, we propose an approach with theoretical guarantees.

Learning-augmented algorithms (also called \emph{algorithms with predictions}) are the subject of a recent line of research that aims to improve over worst-case bounds of classical algorithms by using possibly imperfect predictions. So far, the main focus in this area was on improving competitive ratios for online problems. Dinitz et al.~\cite{Dinitz21} took a first step to explore improving the running times of offline problems. They gave an algorithm for the weighted bipartite matching problem that uses a learned dual solution to improve over the running time of the classic Hungarian algorithm. Our approach draws inspiration from their work, but it differs significantly in two aspects. First, we learn a primal, not a dual solution. Second, we impose an additional restriction on the learned solution, namely the flow conservation property. This restriction makes our learning problem harder and the subsequent algorithmic problem easier. In Section~\ref{sec:related} we discuss these differences in greater depth.

\subsection{Our results}

We propose a framework for speeding up maximum flow computation by using predicted flow values. Here, by \emph{prediction} we mean a flow, which satisfies the flow conservation property, but does not necessarily respect the edge capacities of the actual instance (since these were unknown at the time of learning). We present an algorithm that, given an $m$-edge flow network with edge capacities $c \in \mathbb{Z}_{\geqslant 0}^m$, and a predicted flow $f \in \mathbb{Z}_{\geqslant 0}^m$, computes a maximum flow $f^*(c) \in \mathbb{Z}_{\geqslant 0}^m$ in $O(m \eta)$ time, where $\eta=||f-f^*(c)||_1=\sum_{e\in E}|f(e) - f^*(c)(e)|$ is the $\ell_1$ error of the prediction. Moreover, we prove that, given an oracle access to a (joint) distribution over edge capacities, it is possible to efficiently PAC-learn a prediction minimizing the expected $\ell_1$ error over that distribution.

To formally state our results, let us first define the maximum flow problem and related concepts.

\begin{definition}
\label{def:flow}
Given a directed graph $G=(V, E)$, \emph{source} and \emph{sink} nodes $s, t \in V$, and nonnegative integral edge \emph{capacities} $c: E \to \mathbb{Z}_{\geqslant 0}$, the \emph{maximum flow problem} asks to find a function $f : E \to \mathbb{Z}_{\geqslant 0}$, assigning nonnegative integral \emph{flow} to the edges, that satisfies
\begin{itemize}
\item \emph{capacity constraints}: $\forall_{e \in E} \, f(e) \leqslant c(e)$, and
\item \emph{flow conservation}: $\forall_{v \in V \setminus \{s,t\}} \, \sum_{(u,v) \in E} f(u,v) = \sum_{(v,u) \in E} f(v, u)$,
\end{itemize}
and maximizes the \emph{flow value} defined as $\val{f} = \sum_{(s,u) \in E} f(s,u)$.

We denote by $f^*(c)$ a maximum flow for given capacities $c$.
\end{definition}

Let us note that we have made the decision to focus on the integral version of the problem for two reasons. First, in many applications edge capacities are integral anyway, and hence there always exists an integral solution as well, see, e.g.,~\cite{AhujaMO93}. Second, the error measure we work with, namely the $\ell_1$ distance, is meaningless if one allows arbitrary scaling without changing the problem, as it would the case for rational edge capacities.

In Section~\ref{sec:warmstart}, we prove the following theorem giving an algorithm that can be seen as the Ford-Fulkerson method with a warm start.

\begin{restatable}{theorem}{thmwarmstart}
\label{thm:warmstart}
Given a directed graph $G = (V,E)$, source and sink $s, t \in V$, edge capacities $c: E \to \mathbb{Z}_{\geqslant 0}$, and a predicted flow function $f: E \to \mathbb{Z}_{\geqslant 0}$ satisfying flow conservation, one can compute a maximum $(s,t)$-flow in $G$, in time
\[ O\big(|E| \cdot ||f - f^*(c)||_1\big). \]
\end{restatable}

Note that the above bound holds simultaneously for every maximum flow $f^*(c)$, which might not be unique. In other words, the prediction is good if it is close to at least one optimal solution.

One of the sought-after properties of learning-augmented algorithms is \emph{robustness}, i.e., retaining worst-case guarantees of classic algorithms even for arbitrarily bad predictions. However, in the case of running time bounds, robustness comes essentially for free (up to a multiplicative factor of $2$, vanishing in the asymptotic notation). Indeed, one can always run step-by-step an algorithm with predictions alongside the fastest known classic algorithm, stopping when either of them stops. Therefore, Theorem~\ref{thm:warmstart} paired with the recent $O(|E|^{1+o(1)})$ time algorithm for the maximum flow problem~\cite{ChenKLPPS22} actually leads to a robust learning-augmented algorithm with running time
\[O\big(|E| \cdot \min\{||f - f^*(c)||_1, |E|^{o(1)}\}\big).\]

Now we want to argue that predictions required by the above algorithm can be efficiently learned, in a PAC-learning sense. We assume that the underlying graph, as well as the choice of the source and sink nodes, are fixed. (This assumption is almost without loss of generality, because one can take the underlying graph to be a clique, with capacities zero for nonexistent edges; that, however, may cause a running time overhead, because of the increased number of edges.) Our goal is to prove that, given a joint distribution over edge capacities, we can efficiently learn a flow approximately minimizing the expected $\ell_1$ error over that distribution. We do it in two steps. First, in Section~\ref{sec:learning}, we prove Theorem~\ref{thm:learning}, giving an algorithm that finds an optimal flow prediction for a given set of samples. Next, in Section~\ref{sec:sample}, we prove Theorem~\ref{thm:sample}, arguing that, assuming a sufficient number of samples, such optimal flow for samples is approximately optimal for the whole distribution.

\begin{restatable}{theorem}{thmlearning}
\label{thm:learning}
Given a directed graph $G = (V, E)$, with source and sink $s, t \in V$, and a collection of $k$ lists of edge capacities $c_1, c_2, \ldots, c_k \in \mathbb{Z}_{\geqslant 0}^E$, one can find an optimal integral flow prediction for this collection, i.e.,
\[\hat{f} = \argmin \bigg\{ \frac{1}{k} \sum_{i \in [k]} ||f - f^*(c_i)||_1 \;\big|\; f : E \to \mathbb{Z}_{\geqslant 0} \textnormal{ satisfying flow conservation} \bigg\},\]
in time $O((k \cdot |E|)^{1+o(1)})$.
\end{restatable}

\begin{restatable}{theorem}{thmsample}
\label{thm:sample}
Let $G = (V, E)$ be a directed graph, with source and sink $s, t \in V$, and let $c_1, c_2, \ldots, c_k \in \mathbb{Z}_{\geqslant 0}^E$, for $k = \Theta(c_{\max}^2 |E|^3 \log (c_{\max}|E|))$, be independent samples from a distribution $\mathcal{D}$, where $c_{\max} = \max_{c \in\supp(\mathcal{D}), e \in E} c(e)$.
Let $\hat{f} \in \mathbb{Z}_{\geqslant 0}^E$ be an optimal flow prediction for this collection of samples, as in Theorem~\ref{thm:learning}. Then, with high probability over the choice of the samples, the expected $\ell_1$ error of $\hat{f}$ over $\mathcal{D}$ is approximately minimum possible, i.e.,
\[
\mathbb{E}_{c \sim \mathcal{D}} ||\hat{f} - f^*(c)||_1 \leqslant
\min_f \mathbb{E}_{c \sim \mathcal{D}} ||f - f^*(c)||_1 + O(1),
\]
where the minimum is taken over functions $f \in \mathbb{Z}_{\geqslant 0}^E$ satisfying the flow conservation property.
\end{restatable}

\subsection{Related work}
\label{sec:related}

\paragraph{Maximum flow algorithms.}

There are numerous algorithms for the maximum flow problem. The Ford-Fulkerson method~\cite{FordF56} is a starting point for many of them, and its vanilla version runs in weakly polynomial $O(|E|\cdot\val{f^*(c)})$ time for integral edge capacities. The strongly polynomial time algorithms, which also work for rational edge capacities, can be roughly split into three groups: augmenting paths algorithms (e.g.,~\cite{EdmondsK72, Dinic70}), push-relabel algorithms (e.g.,~\cite{GoldbergT86}), and pseudoflow algorithms (e.g.,~\cite{Hochbaum08}). Each of these groups contains algorithms with running time $\widetilde{O}(|V|\cdot|E|)$ that are widely used in practice, see, e.g.,~\cite{BoykovK04,FishbainHM16}. A long line of research on Laplacian solvers and interior-point methods, initiated by~\cite{SpielmanT04}, culminated recently with a (weakly polynomial) near-linear $O(|E|^{1+o(1)})$ time algorithm~\cite{ChenKLPPS22}.

In the light of this new development, it may seem that our learning-augmented algorithm is only relevant for very small prediction errors, namely $||f-f^*(c_i)||_1 \leqslant |E|^{1+o(1)}$. However, at this point it is not yet clear if the new near-linear time algorithm will lead to practical developments.\footnote{See \url{https://codeforces.com/blog/entry/100510} for a relevant discussion with an author of~\cite{ChenKLPPS22}.}

\paragraph{Learning-augmented algorithms.}

The idea of using predictions to improve performance of algorithms is not a new one, see, e.g.,~\cite{MahdianNS07}. However, the recent systematic study of such methods -- under the umbrella term on \emph{learning-augmented algorithms}, or simply \emph{algorithms with predictions} -- seems to have started with the works of Lykouris and Vassilvitskii~\cite{LykourisV21}, and Purohit, Svitkina, and Kumar~\cite{PurohitSK18}. Since then, the field developed rapidly, see~\cite{MitzenmacherV20} for a survey. So far, the majority of the works focus on online algorithms, where predictions help reduce uncertainty about the yet unseen part of the instance. There are, however, also works on, e.g., data structures~\cite{KraskaBCDP18}, streaming algorithms~\cite{HsuIKV19}, and sublinear algorithms~\cite{EdenINRSW21}. Apart from a simple example of binary search~\cite{LykourisV21}, until recently there were no works on improving algorithms running times using predictions. This has changed with the work of Dinitz et al.~\cite{Dinitz21}, and the recent followup work of Chen et al.~\cite{ChenSVZ22}.

\paragraph{Learning-augmented weighted bipartite matching.}

A direct inspiration for our approach is the work of Dinitz et al.~\cite{Dinitz21}. They study the maximum weighted bipartite matching problem, and propose to predict the dual\footnote{Recall that the matching problem can be formulated as a linear program, and every linear program has a corresponding dual program.} solution. They give a learning-augmented algorithm that solves the matching problem in $O(|E|\sqrt{|V|} \cdot \min\{\eta, \sqrt{|V|}\})$ time, where $\eta$ is the $\ell_1$ error of the predicted dual solution -- our Theorem~\ref{thm:warmstart} is an analogue of that result. They also show that, given an oracle access to a joint distribution over edge weights, one can efficiently learn a prediction minimizing the expected $\ell_1$ error over the distribution -- our Theorems~\ref{thm:learning} and~\ref{thm:sample} are together an analogue of that result.

The most apparent difference between their approach and ours is that they use a predicted dual solution and we use a predicted primal solution. The reason they state for focusing on the dual solution is that the primal solution is very volatile to small changes in the input. Let us note that this argument clearly applies to weighted problems (in particular, e.g., to the minimum cost flow problem) but it is not clear if it also applies to the maximum flow problem. Moreover, it is also not clear if the dual solution is indeed less volatile, even for weighted problems.

The second important difference is that they do not impose any constraints on predictions, while we require that the predicted solution satisfies the flow conservation property. This difference has the following consequences. First, their learning algorithm can be very simple -- the optimal prediction is just a coordinate-wise median over the solutions for the samples -- while we need to solve the minimum cost flow problem instead. Second, turning a prediction into a feasible solution is also harder for us, as we want to maintain the flow conservation property. On the other hand, once we have a feasible solution, the remaining part of our maximum flow algorithm is simple and easy to analyse, in contrast with their tailored primal-dual analysis for the analogous part of their algorithm.

In a recent independent work Chen et al.~\cite{ChenSVZ22} improve the running time of Dinitz et al.~for the matching problem, and extend their framework to a couple of other problems: the negative weights single-source shortest paths problem, the degree-constrained subgraph problem, and the minimum cost 0-1 flow problem. For all these problems they use predicted dual solutions. They also propose general learnability theorems, which imply what we prove in Appendix~\ref{sec:altsample} (see also a discussion below the proof of Lemma~\ref{lem:sample} for a comparison of these results with Theorem~\ref{thm:sample}).

\section{Warm-starting Ford-Fulkerson}
\label{sec:warmstart}

\thmwarmstart*

\begin{proof}

At first, the predicted flow $f$ does not necessarily satisfy the capacity constraints imposed by $c$, i.e., for some edges $e \in E$ it might happen that $f(e) > c(e)$. The algorithm consists of two steps. In the first step, it turns $f$ into $\bar{f}$ that satisfies the capacity constraints, while maintaining the flow conservation property. In other words, $\bar{f}$ is a feasible flow. Then, in the second step, the algorithm augments $\bar{f}$ to an optimal flow.

\vspace{-1em}\paragraph{First step: feasibility.} Recall that every integral flow decomposes into cycles and $s$-$t$ paths\footnote{I.e., there exists a collection $p_1, \ldots, p_k$ such that each $p_i$ is either a (simple) cycle or a (simple) $s$-$t$ path in~$G$, and $f(e) = \#\{i \in [k] \mid e \in p_i\}$ for every edge $e \in E$.} (see, e.g., \cite[Theorem~3.5]{AhujaMO93}).
The algorithm initializes $\bar{f}=f$. While there is an edge $e \in E$ with $\bar{f}(e) > c(e)$, the algorithm uses, e.g., depth-first search to find a cycle or an $s$-$t$ path containing $e$ (at least one of them is guaranteed to exist because of the integral flow decomposition), and decreases the flow $\bar{f}$ along this cycle/path by one unit. This keeps the invariant that $\bar{f}$ satisfies the flow conservation property. When the process is done, $\bar{f}$ satisfies also all the capacity constraints.

\vspace{-1em} \paragraph{Second step: optimization.} Now, the algorithm constructs the \emph{residual network} with respect to $\bar{f}$, i.e., the flow network $G_{\bar{f}} = (V, E_{\bar{f}})$ with edge set $E_{\bar{f}} = \{ (u, v) \mid (u, v) \in E \text { or } (v, u) \in E\}$ and residual capacities $c_{\bar{f}}(u, v) = (c(u, v) - f(u, v)) + f(v, u)$. Here, for notational simplicity, we assume that $c(u, v) = f(u, v) = 0$ if $(u, v) \not\in E$. 
Then, the algorithm runs the Ford-Fulkerson method~\cite{FordF56} on $G_{\bar{f}}$ to find a maximum flow $f^*(c_{\bar{f}})$ in time $O(|E| \cdot \val{f^*(c_{\bar{f}})})$.
Finally, $\bar{f} + f^*(c_{\bar{f}})$ is a maximum flow for the original edge capacities $c$, see, e.g., \cite[Property~2.6]{AhujaMO93}.

\vspace{-1em} \paragraph{Running time analysis.} Let $\delta = \sum_{e \in E} \max \{f(e) - c(e), 0\}$ be the total amount by which the flow prediction violates the capacity constraints. The algorithm makes at most $\delta$~iterations in the first step, and each iteration decreases the flow value $\val{\bar{f}}$ by at most one. We conclude that the first step runs in $O(|E|\cdot \delta)$ time, and that $\val{f} - \val{\bar{f}} \leqslant \delta$.

The second step of the algorithm runs in time 
\begin{align*}
O(|E| \cdot \val{f^*(c_{\bar{f}})})
 & = O(|E| \cdot (\val{f^*(c)} - \val{\bar{f}})) \\
 & = O(|E| \cdot ((\val{f^*(c)} - \val{f}) + (\val{f} - \val{\bar{f}}))).
\end{align*}
Let $\eta = ||f-f^*(c)||_1$ denote the prediction error. It is easy to see that $| \val{f^*(c)} - \val{f} | \leqslant \eta$, and that $\delta \leqslant \eta$, so, in particular, $\val{f} - \val{\bar{f}} \leqslant \delta \leqslant \eta$. Therefore, the running time of both steps of the algorithm can be bounded by $O(|E| \cdot \eta)$.
\end{proof}

\subsection{Alternative variant of the first step}

In this section we give an alternative variant of the first step of the above algorithm. The asymptotic running time remains the same, but, as we explain towards the end of this section, the alternative algorithm might be more efficient in practice.

Consider graph $\widetilde G = (V, \widetilde E)$ with $\widetilde E = \{(v, u) \in V \times V \mid (u, v) \in E\}$, i.e., a copy of $G$ with reversed edges. Set capacities to $\tilde c(v, u) = f(u, v)$. Note that the first step of the original algorithm essentially finds an integral $t$-$s$ flow $\tilde f$ in $\widetilde G$ such that
\begin{enumerate}[label=(\roman*),nosep]
    \item if $f(u, v) > c(u, v)$, then $\tilde f(v, u) \geqslant f(u, v) - c(u, v)$, for every $(u, v) \in E$;
    \item $\val{\tilde f} \leqslant \delta$.
\end{enumerate}
At the end of the first step $\bar f(u, v) = f(u, v) - \tilde f(v, u)$. In this section we give an alternative way to compute such $\tilde f$.

Add to $\widetilde G$ edge $(s, t)$, and set $\tilde c(s, t) = \delta$. Now, the problem of finding a $t$-$s$ flow satisfying (i) and (ii) becomes the problem of finding a \emph{circulation}\footnote{A circulation is defined similarly to a flow. The only exception is that there are no designated source and sink nodes, and hence the flow conservation property has to be satisfied for all the nodes of the graph (see, e.g., \cite[Section~1.2]{AhujaMO93}).} satisfying (i). This problem -- of finding a circulation with lower bounds -- can be reduced to a problem of finding a maximum flow (without lower bounds) in a graph with the maximum flow value equal to the sum of all lower bounds, see, e.g., \cite[Section~6.7]{AhujaMO93}. The reduction works as follows.

First, add to $\widetilde G$ two new nodes $\tilde s$ and $\tilde t$. Next, for each edge $e = (u, v) \in E$ that violates the capacity constraint let $\delta_e = f(u,v) - c(u,v) > 0$ be the excess flow for this edge; add to $\widetilde G$ two edges, $(\tilde s, u)$ and $(v, \tilde t)$, set their capacities to $\tilde c(\tilde s, u) = \tilde c(v, \tilde t) = \delta_e$, and decrease the capacity of edge $(v, u)$ by $\delta_e$, so that $\tilde c(v, u) = c(u, v)$. This ends the description of the graph constructed in the reduction.

Note that the total capacity of edges leaving $\tilde s$ equals the total capacity of edges entering $\tilde t$ equals $\delta = \sum_{e \in E} \max \{f(e) - c(e), 0\}$. As we will see in a moment, the existence of a flow saturating these edges is equivalent to the existence of a circulation satisfying the lower bounds -- which is guaranteed to exist because the original first step of the algorithm finds such a circulation.

After constructing $\widetilde G$ as above, the alternative first step proceeds as follows. The algorithm computes a maximum $\tilde s$-$\tilde t$ flow $\tilde f$ in $\widetilde G$, using Ford-Fulkerson method. Then, for each edge $e = (u, v) \in E$ that violates the capacity constraint (in the original graph $G$), the algorithm first removes the saturated edges $(\tilde s, u)$ and $(v, \tilde t)$ from $\widetilde G$. Note that now nodes $u$ and $v$ do not satisfy the flow conservation property, namely node $v$ has an excess of $\delta_e$ units of incoming flow and node $u$ has a deficit of $\delta_e$ units of incoming flow. The algorithm restores the flow conservation property by increasing flow $\tilde f(v, u)$ by $\delta_e$ units, and therefore it ensures that the lower bound for this edge is satisfied. This procedure essentially proves the equivalence of the existence of a flow saturating sink and source edges and the existence of a suitable circulation.

This ends the description of the alternative algorithm. Let us analyze its running time. Graph $\widetilde G$ has $O(|E|)$ edges and can be constructed in $O(|E|)$ time. Since $\val{\tilde f} = \delta$, the Ford-Fulkerson method runs in $O(|E|\cdot \delta)$ time. Finally, transforming $\tilde f$ to $\bar f$ takes $O(|E|)$ time. Therefore, the total running time of $O(|E| \cdot \delta)$ remains unchanged compared to the original first step of the algorithm. However, the alternative algorithm differs from the original one in that all the computations that take more than $O(|E|)$ time can be delegated to one of many available highly optimized implementations of maximum flow algorithms. 

Finally, we remark that a similar trick -- for handling edges initialized with a flow exceeding their capacities -- was already proposed, albeit without provable running time guarantees, in the context of repeatedly solving similar minimum cut instances in a computer vision application~\cite{KohliT05}. That trick however only allows computing the maximum flow value and a corresponding minimum cut, but not the flow itself.

\section{Learning an optimal prediction}
\label{sec:learning}

\thmlearning*

\begin{proof}
The first step of the learning algorithm is to compute a maximum flow $f^*(c_i)$ for each $i \in [k]$. Using the recent near-linear time algorithm~\cite{ChenKLPPS22}, this step can be completed in $O(k\cdot|E|^{1+o(1)})$ time in total.

Now, the goal is to find an integral flow $f$ (satisfying the flow conservation property) that minimizes $\sum_{i \in [k]} \sum_{e \in E} |f(e) - f^*(c_i)(e)| = \sum_{e \in E} \sum_{i \in [k]} |f(e) - f^*(c_i)(e)|$. For an edge $e \in E$, let $\cost_e(x) = \sum_{i \in [k]} |x - f^*(c_i)(e)|$ denote the contribution of $f(e)$ to the minimization objective, which now can be written simply as $\sum_{e \in E} \cost_e(f(e))$.

Let us analyse how the function $\cost_e(x)$ behaves. Let $x_1 \leqslant x_2 \leqslant \cdots \leqslant x_k$ denote the sorted elements of the (multi-)set $\{f^*(c_1)(e), f^*(c_2)(e), \ldots, f^*(c_k)(e)\}$. Clearly, $\cost_e(0) = \sum_{i \in [k]} x_i$. For $x \in [0, x_1]$, the contribution $\cost_e(x)$ is a decreasing linear function with slope $-k$. For $x \in [x_1, x_2]$, the slope is $-k+2$. More generally, for $x \in [x_i, x_{i+1}]$ the slope is $2i - k$, because increasing the flow by $\delta$ increases also by $\delta$ each of the first $i$ summands, and decreases by the same amount each of the remaining $(k-i)$ summands in the sum $\sum_{i \in [k]} |x - x_i|$. Hence, $\cost_e(x)$ is piecewise linear and convex, and the overall goal is to find a flow minimizing a separable piecewise linear convex cost function.

The above problem can be reduced to the standard minimum cost flow problem~\cite[Chapter 14]{AhujaMO93}. The reduction works as follows. For notational simplicity, let $x_0=0$ and $x_{k+1}=+\infty$. Replace each edge $e \in E$ with $k+1$ parallel edges $e_0, \ldots, e_k$, and let edge $e_i$ have capacity $x_{i+1}-x_i$ and cost (of sending one unit of flow) equal to $2i-k$. It is easy to observe that any optimal solution to the minimum cost flow problem in the constructed multigraph uses some prefix of the cheapest parallel edges for each $e \in E$, and the total cost of such prefix behaves exactly like $\cost_e$. Since all the introduced capacities are integral, it is guaranteed that there exists an optimal integral solution. The multigraph has $(k+1)\cdot|E|$ edges, hence the minimum cost flow can be found in $O((k\cdot|E|)^{1+o(1)})$ time~\cite{ChenKLPPS22}.
\end{proof}

\section{Sample complexity}
\label{sec:sample}

\thmsample*

For a flow prediction $f \in \mathbb{Z}_{\geqslant 0}^E$, let us use
\begin{align*}
    \cost_{c_1,\ldots,c_k}(f) &= \frac{1}{k} \sum_{i \in [k]} || f - f^*(c_i)||_1, \\
    \cost_{\mathcal{D}}(f) &= \mathbb{E}_{c \sim \mathcal{D}} || f - f^*(c)||_1
\end{align*}
to denote the $\ell_1$ error of $f$ on the samples and on the distribution, respectively. We will use Hoeffding's inequality to prove that the number of samples in Theorem~\ref{thm:sample} is large enough for $\cost_{c_1,\ldots,c_k}(f)$ to be a good approximation of $\cost_{\mathcal{D}}(f)$, with high probability for all $f$'s simultaneously.

\begin{theorem}[Hoeffding's inequality~\cite{Hoeffding63}]
Let $X_1, \ldots, X_k$ be independent random variables with values from $0$ to $U$, and let $S = X_1 + \cdots + X_k$ denote their sum. Then, for all $t > 0$,
\[P\big( | S - \mathbb{E} S | \geqslant t \big) \leqslant 2 \cdot \exp(-2t^2/kU^2).\]
\end{theorem}

To use the inequality, first we need a bound on the values of the considered functions.

\begin{lemma}
\label{lem:fnorm}
Any optimal flow prediction $f$ satisfies $||f||_1 \leqslant 2 c_{\max} |E|$.
\end{lemma}

Let us note that Lemma~\ref{lem:fnorm} is actually nontrivial. Even though $||f^*(c_i)||_\infty \leqslant c_{\max}$ for every $i \in [k]$, it may happen that $||f||_\infty > c_{\max}$ because of the flow conservation constraint, e.g., when multiple disjoint paths end at a single node and force a single edge going out of that node to have a flow larger than $c_{\max}$.

\begin{proof}[Proof of Lemma~\ref{lem:fnorm}]
For every $c \in \supp(\mathcal{D})$, we have $||c||_1 \leqslant c_{\max}|E|$, and, since $0 \leqslant f^*(c) \leqslant c$, then also $||f^*(c)||_1 \leqslant c_{\max}|E|$. Moreover, by the triangle inequality, $||f-f^*(c)||_1 + ||f^*(c)||_1 \geqslant ||f||_1$, and thus $||f-f^*(c)||_1 \geqslant ||f||_1 - c_{\max}|E|$. If $||f||_1 > 2c_{\max}|E|$, then $||f-f^*(c)||_1 > c_{\max}|E|$ for every $c \in \supp(\mathcal{D})$, and thus also $\cost_{\mathcal{D}}(f) = \mathbb{E}_{c \sim \mathcal{D}} ||f-f^*(c)||_1 > c_{\max}|E|$.

At the same time, if we consider the all-zero vector as a flow prediction, we have $||0-f^*(c)||_1 = ||f^*(c)||_1 \leqslant c_{\max}|E|$, for every $c \in \supp(\mathcal{D})$, and thus also $\cost_{\mathcal{D}}(0) = \mathbb{E}_{c \sim \mathcal{D}} ||0-f^*(c)||_1 \leqslant c_{\max}|E|$. It follows that $f$ could not be optimal if $||f||_1 > 2 c_{\max} |E|$.
\end{proof}

Now we are ready to apply Hoeffding's inequality in order to prove the following lemma.

\begin{restatable}{lemma}{lemsample}\label{lem:sample}
With high probability over the choice of the samples, for all $f \in \mathbb{Z}_{\geqslant 0}^E$ satisfying the flow conservation property and such that $||f||_1 \leqslant 2c_{\max}|E|$ it holds that
\[| \cost_{c_1,\ldots,c_k}(f) - \cost_{\mathcal{D}}(f) | \leqslant 1.\]
\end{restatable}

\begin{proof}
For a fixed $f$, satisfying the conditions of the lemma, let $X_i = \frac{1}{k} ||f - f^*(c_i)||_1$. We have that $||f - f^*(c_i)||_1 \leqslant ||f||_1 + ||f^*(c_i)||_1 \leqslant (2+1) \cdot c_{\max} |E|$, so the random variable $X_i$ has values from $0$ to $3 c_{\max}|E| / k]$. Clearly, $\cost_{c_1,\ldots,c_k}(f) = X_1 + \cdots + X_k$, and $\mathbb{E} \cost_{c_1,\ldots,c_k}(f)= \cost_{\mathcal{D}}(f)$. Applying Hoeffding's inequality, with $t=1$, we get that
\begin{align*}
P\big( | \cost_{c_1,\ldots,c_k}(f) - \cost_{\mathcal{D}}(f) | \geqslant 1 \big) & \leqslant 2 \cdot \exp\bigg(-\frac{2}{k \cdot (3 c_{\max}|E| / k)^2}\bigg) \\
& = \exp(-\Theta(\nicefrac{k}{(c_{\max}|E|)^2})) \\
&= \exp(-\Theta(|E|\log(c_{\max}|E|))).
\end{align*}

Let $\mathcal{F}$ denote the set of all $f$'s satisfying the conditions of the lemma. We can upper-bound the number of such $f$'s by $|\mathcal{F}| \leqslant (2c_{\max}|E| + 1)^{|E|} = \exp(\Theta(|E| \log (c_{\max}|E|)))$. To finish the proof, note that
\begin{align*}
|\mathcal{F}| \cdot \poly(c_{\max}|E|) & \leqslant \exp(\Theta(|E| \log (c_{\max}|E|))) \cdot \exp(\Theta( \log (c_{\max}|E|))) \\ &= \exp(\Theta(|E| \log (c_{\max}|E|))),
\end{align*}
and hence we can take the union bound to conclude that with high probability it holds that $| \cost_{c_1,\ldots,c_k}(f) - \cost_{\mathcal{D}}(f) | \leqslant 1$ for all $f \in \mathcal{F}$ simultaneously.
\end{proof}

Let us remark that the above proof of Lemma~\ref{lem:sample} crucially relies on the fact that the set of possible optimal flow predictions is finite -- because they are integral and bounded -- and therefore we can use the union bound. Dinitz et al.~\cite[Section~3.3 in their supplemental material]{Dinitz21} give a proof of an analogous result regarding learning optimal dual solution for the weighted bipartite matching problem. Their proof is more complex than ours, it uses the notion of pseudo-dimension, but thanks to that it works also for fractional predictions. In Appendix~\ref{sec:altsample} we give a proof of alternative version of Lemma~\ref{lem:sample}, modelled after the proof of Dinitz et al., that generalizes to fractional flows but looses a small factor $\log |E|$ in the sample complexity.

With Lemma~\ref{lem:sample} at hand, it takes a standard argument to prove Theorem~\ref{thm:sample}.

\begin{proof}[Proof of Theorem~\ref{thm:sample}]
Let $\hat{f}$ and $\tilde{f}$ be optimal flow predictions for the samples and for the whole distribution, respectively. By Lemma~\ref{lem:fnorm}, $||\hat{f}||_1, ||\tilde{f}||_1 \leqslant 2c_{\max}|E|$, and hence
Lemma~\ref{lem:sample} applies.
Note that it is crucial that Lemma~\ref{lem:sample} holds with high probability for all $f$'s, because $\hat{f}$ is chosen after the samples are drawn from $\mathcal{D}$.
We finish the proof with the following chain of inequalities.
\[
\cost_{\mathcal{D}}(\hat{f}) 
\underset{\mathclap{\substack{\uparrow\\\text{Lemma~\ref{lem:sample}}}}}{\leqslant}
\cost_{c_1,\ldots,c_k}(\hat{f}) + 1
\underset{\mathclap{\substack{\uparrow\\[-2pt]\text{because $\hat{f}$ is optimal for $c_1, \ldots, c_k$}}}}{\leqslant}
\cost_{c_1,\ldots,c_k}(\bar{f}) + 1
\underset{\mathclap{\substack{\uparrow\\\text{Lemma~\ref{lem:sample}}}}}{\leqslant}
\cost_{\mathcal{D}}(\bar{f}) + 2.
\]
\end{proof}

\section{Limitations and open problems}
\label{sec:limitations}

\paragraph{Representation error.}
We do prove that a prediction with a small $\ell_1$ error can be used to speed up maximum flow computation, and that given a distribution over flow networks one can learn a prediction minimizing the $\ell_1$ error. However, we do not answer the question of what makes a distribution have such a minimum that is actually small. There seems to be no standard approach to address this type of question, and the related works~\cite{Dinitz21,ChenSVZ22} do not address it either.

\paragraph{Dropping flow conservation constraints.}
One could consider a similar framework to ours but without the requirement that the predicted solution has to satisfy the flow conservation property. That would make a) the learning algorithm simpler (it would be sufficient to output the coordinate-wise median), b) the minimum $\ell_1$ error for a distribution smaller, and c) the capacity constraint fixing step of the learning-augmented maximum-flow algorithm simpler (it would be sufficient to clip the predicted flows to the actual capacities). However, having no flow conservation guarantee at the warm start, the second step of the algorithm would have to deal with both nodes with excess and deficit flow. The pseudoflow algorithm~\cite{Hochbaum08} does work in such a setting -- so it seems a promising starting point for a learning-augmented algorithm in this modified framework -- but we were not able to analyse its performance in terms of the prediction error.

\section*{Acknowledgments}
We would like to thank Alexandra Lassota, Sai Ganesh Nagarajan, and Moritz Venzin for helpful discussion.

\bibliographystyle{plainurl}
\bibliography{main}

\begin{thebibliography}{10}

\bibitem{AhujaMO93}
Ravindra~K. Ahuja, Thomas~L. Magnanti, and James~B. Orlin.
\newblock {\em Network flows -- theory, algorithms and applications}.
\newblock Prentice Hall, 1993.

\bibitem{Anthony02}
Martin Anthony and Peter~L. Bartlett.
\newblock {\em Neural Network Learning -- Theoretical Foundations}.
\newblock Cambridge University Press, 2002.
\newblock URL:
  \url{http://www.cambridge.org/gb/knowledge/isbn/item1154061/?site\_locale=en\_GB}.

\bibitem{Bartlett98}
Peter~L. Bartlett, Vitaly Maiorov, and Ron Meir.
\newblock Almost linear {VC} dimension bounds for piecewise polynomial
  networks.
\newblock In {\em Advances in Neural Information Processing Systems 11, {[NIPS}
  Conference, 1998]}, pages 190--196. The {MIT} Press, 1998.
\newblock URL:
  \url{http://papers.nips.cc/paper/1515-almost-linear-vc-dimension-bounds-for-piecewise-polynomial-networks}.

\bibitem{BoykovK04}
Yuri Boykov and Vladimir Kolmogorov.
\newblock An experimental comparison of min-cut/max-flow algorithms for energy
  minimization in vision.
\newblock {\em {IEEE} Trans. Pattern Anal. Mach. Intell.}, 26(9):1124--1137,
  2004.
\newblock \href {https://doi.org/10.1109/TPAMI.2004.60}
  {\path{doi:10.1109/TPAMI.2004.60}}.

\bibitem{ChenSVZ22}
Justin~Y. Chen, Sandeep Silwal, Ali Vakilian, and Fred Zhang.
\newblock Faster fundamental graph algorithms via learned predictions.
\newblock {\em CoRR}, abs/2204.12055, 2022.
\newblock \href {http://arxiv.org/abs/2204.12055} {\path{arXiv:2204.12055}},
  \href {https://doi.org/10.48550/arXiv.2204.12055}
  {\path{doi:10.48550/arXiv.2204.12055}}.

\bibitem{ChenKLPPS22}
Li~Chen, Rasmus Kyng, Yang~P. Liu, Richard Peng, Maximilian Probst~Gutenberg,
  and Sushant Sachdeva.
\newblock Maximum flow and minimum-cost flow in almost-linear time, 2022.
\newblock \href {https://doi.org/10.48550/ARXIV.2203.00671}
  {\path{doi:10.48550/ARXIV.2203.00671}}.

\bibitem{Dinitz21}
Michael Dinitz, Sungjin Im, Thomas Lavastida, Benjamin Moseley, and Sergei
  Vassilvitskii.
\newblock Faster matchings via learned duals.
\newblock In {\em Advances in Neural Information Processing Systems},
  volume~34, pages 10393--10406. Curran Associates, Inc., 2021.
\newblock URL:
  \url{https://papers.nips.cc/paper/2021/hash/5616060fb8ae85d93f334e7267307664-Abstract.html}.

\bibitem{Dinic70}
Yefim Dinitz.
\newblock Algorithm for solution of a problem of maximum flow in networks with
  power estimation.
\newblock {\em Soviet Math. Dokl.}, 11:1277--1280, 1970.

\bibitem{EdenINRSW21}
Talya Eden, Piotr Indyk, Shyam Narayanan, Ronitt Rubinfeld, Sandeep Silwal, and
  Tal Wagner.
\newblock Learning-based support estimation in sublinear time.
\newblock In {\em 9th International Conference on Learning Representations,
  {ICLR} 2021}. OpenReview.net, 2021.
\newblock URL: \url{https://openreview.net/forum?id=tilovEHA3YS}.

\bibitem{EdmondsK72}
Jack~R. Edmonds and Richard~M. Karp.
\newblock Theoretical improvements in algorithmic efficiency for network flow
  problems.
\newblock {\em J. {ACM}}, 19(2):248--264, 1972.
\newblock \href {https://doi.org/10.1145/321694.321699}
  {\path{doi:10.1145/321694.321699}}.

\bibitem{FishbainHM16}
Barak Fishbain, Dorit~S. Hochbaum, and Stefan M{\"{u}}ller.
\newblock A competitive study of the pseudoflow algorithm for the minimum s-t
  cut problem in vision applications.
\newblock {\em J. Real Time Image Process.}, 11(3):589--609, 2016.
\newblock \href {https://doi.org/10.1007/s11554-013-0344-3}
  {\path{doi:10.1007/s11554-013-0344-3}}.

\bibitem{FordF56}
L.~R. Ford and D.~R. Fulkerson.
\newblock Maximal flow through a network.
\newblock {\em Canadian Journal of Mathematics}, 8:399--404, 1956.
\newblock \href {https://doi.org/10.4153/CJM-1956-045-5}
  {\path{doi:10.4153/CJM-1956-045-5}}.

\bibitem{GoldbergT86}
Andrew~V. Goldberg and Robert~Endre Tarjan.
\newblock A new approach to the maximum flow problem.
\newblock In {\em Proceedings of the 18th Annual {ACM} Symposium on Theory of
  Computing, 1986}, pages 136--146. {ACM}, 1986.
\newblock \href {https://doi.org/10.1145/12130.12144}
  {\path{doi:10.1145/12130.12144}}.

\bibitem{Hochbaum08}
Dorit~S. Hochbaum.
\newblock The pseudoflow algorithm: {A} new algorithm for the maximum-flow
  problem.
\newblock {\em Oper. Res.}, 56(4):992--1009, 2008.
\newblock \href {https://doi.org/10.1287/opre.1080.0524}
  {\path{doi:10.1287/opre.1080.0524}}.

\bibitem{Hoeffding63}
Wassily Hoeffding.
\newblock Probability inequalities for sums of bounded random variables.
\newblock {\em Journal of the American Statistical Association},
  58(301):13--30, 1963.
\newblock \href {https://doi.org/10.1080/01621459.1963.10500830}
  {\path{doi:10.1080/01621459.1963.10500830}}.

\bibitem{HsuIKV19}
Chen{-}Yu Hsu, Piotr Indyk, Dina Katabi, and Ali Vakilian.
\newblock Learning-based frequency estimation algorithms.
\newblock In {\em 7th International Conference on Learning Representations,
  {ICLR} 2019}. OpenReview.net, 2019.
\newblock URL: \url{https://openreview.net/forum?id=r1lohoCqY7}.

\bibitem{JuanB06}
Olivier Juan and Yuri Boykov.
\newblock Active graph cuts.
\newblock In {\em 2006 {IEEE} Computer Society Conference on Computer Vision
  and Pattern Recognition ({CVPR} 2006)}, pages 1023--1029. {IEEE} Computer
  Society, 2006.
\newblock \href {https://doi.org/10.1109/CVPR.2006.47}
  {\path{doi:10.1109/CVPR.2006.47}}.

\bibitem{KohliT05}
Pushmeet Kohli and Philip H.~S. Torr.
\newblock Efficiently solving dynamic markov random fields using graph cuts.
\newblock In {\em 10th {IEEE} International Conference on Computer Vision
  {(ICCV} 2005)}, pages 922--929. {IEEE} Computer Society, 2005.
\newblock \href {https://doi.org/10.1109/ICCV.2005.81}
  {\path{doi:10.1109/ICCV.2005.81}}.

\bibitem{KraskaBCDP18}
Tim Kraska, Alex Beutel, Ed~H. Chi, Jeffrey Dean, and Neoklis Polyzotis.
\newblock The case for learned index structures.
\newblock In {\em Proceedings of the 2018 International Conference on
  Management of Data, {SIGMOD} Conference 2018}, pages 489--504. {ACM}, 2018.
\newblock \href {https://doi.org/10.1145/3183713.3196909}
  {\path{doi:10.1145/3183713.3196909}}.

\bibitem{LykourisV21}
Thodoris Lykouris and Sergei Vassilvitskii.
\newblock Competitive caching with machine learned advice.
\newblock {\em J. {ACM}}, 68(4):24:1--24:25, 2021.
\newblock \href {https://doi.org/10.1145/3447579} {\path{doi:10.1145/3447579}}.

\bibitem{MahdianNS07}
Mohammad Mahdian, Hamid Nazerzadeh, and Amin Saberi.
\newblock Allocating online advertisement space with unreliable estimates.
\newblock In {\em Proceedings 8th {ACM} Conference on Electronic Commerce
  (EC-2007)}, pages 288--294. {ACM}, 2007.
\newblock \href {https://doi.org/10.1145/1250910.1250952}
  {\path{doi:10.1145/1250910.1250952}}.

\bibitem{MitzenmacherV20}
Michael Mitzenmacher and Sergei Vassilvitskii.
\newblock Algorithms with predictions.
\newblock In Tim Roughgarden, editor, {\em Beyond the Worst-Case Analysis of
  Algorithms}, pages 646--662. Cambridge University Press, 2020.
\newblock \href {https://doi.org/10.1017/9781108637435.037}
  {\path{doi:10.1017/9781108637435.037}}.

\bibitem{Morgenstern15}
Jamie Morgenstern and Tim Roughgarden.
\newblock On the pseudo-dimension of nearly optimal auctions.
\newblock In {\em Advances in Neural Information Processing Systems 28: Annual
  Conference on Neural Information Processing Systems 2015}, pages 136--144,
  2015.
\newblock URL:
  \url{https://proceedings.neurips.cc/paper/2015/hash/fbd7939d674997cdb4692d34de8633c4-Abstract.html}.

\bibitem{Pollard84}
David Pollard.
\newblock {\em Convergence of Stochastic Processes}.
\newblock Springer, 1984.
\newblock \href {https://doi.org/https://doi.org/10.1007/978-1-4612-5254-2}
  {\path{doi:https://doi.org/10.1007/978-1-4612-5254-2}}.

\bibitem{PurohitSK18}
Manish Purohit, Zoya Svitkina, and Ravi Kumar.
\newblock Improving online algorithms via {ML} predictions.
\newblock In {\em Advances in Neural Information Processing Systems 31: Annual
  Conference on Neural Information Processing Systems 2018, NeurIPS 2018},
  pages 9684--9693, 2018.
\newblock URL:
  \url{https://proceedings.neurips.cc/paper/2018/hash/73a427badebe0e32caa2e1fc7530b7f3-Abstract.html}.

\bibitem{ShalevShwartz14}
Shai Shalev{-}Shwartz and Shai Ben{-}David.
\newblock {\em Understanding Machine Learning -- From Theory to Algorithms}.
\newblock Cambridge University Press, 2014.

\bibitem{SpielmanT04}
Daniel~A. Spielman and Shang{-}Hua Teng.
\newblock Nearly-linear time algorithms for graph partitioning, graph
  sparsification, and solving linear systems.
\newblock In {\em Proceedings of the 36th Annual {ACM} Symposium on Theory of
  Computing, 2004}, pages 81--90. {ACM}, 2004.
\newblock \href {https://doi.org/10.1145/1007352.1007372}
  {\path{doi:10.1145/1007352.1007372}}.

\end{thebibliography}

\clearpage

\appendix

\section{Sample complexity via pseudo-dimension}
\label{sec:altsample}

In this section we give an alternative proof of a variant of Lemma~\ref{lem:sample}, modelled after a corresponding proof by Dinitz et al.~\cite[Section~3.3 in their supplemental material]{Dinitz21}. Let us recall Lemma~\ref{lem:sample} first.

\lemsample*

Note that $k$ in Lemma~\ref{lem:sample} refers to the number of samples in Theorem~\ref{lem:sample}, i.e., $k = \Theta(c_{\max}^2 |E|^3 \log (c_{\max}|E|))$.
As mentioned in Section~\ref{sec:sample}, the technique of Dinitz et al.~lets us prove a result that also applies to fractional flows but gives a slightly worse (by $\log E$ factor) sample complexity. Specifically, we will prove the following lemma.

\begin{lemma}
\label{lem:altsample}
For $k=\Theta(c_{\max}^2 |E|^3 \log (c_{\max} |E|) \log E)$ samples, with high probability over the choice of the samples, for all $f \in \mathbb{R}_{\geqslant 0}^E$ satisfying the flow conservation property and such that $||f||_1 \leqslant 2c_{\max}|E|$ it holds that
\[| \cost_{c_1,\ldots,c_k}(f) - \cost_{\mathcal{D}}(f) | \leqslant 1.\]
\end{lemma}

To prove the lemma, first let us recall a standard tool from the statistical learning theory, the pseudo-dimension, which is a generalization of the VC-dimension to real-valued functions.

\begin{definition}[cf.~\cite{Pollard84}]
Let $\mathcal{F} \subseteq \mathbb{R}^X$ be a set of real-valued functions from a domain $X$.
We say that a set $S = \{x_1, x_2, \ldots, x_s\} \subseteq X$ is \emph{shattered} by $\mathcal{F}$ if there exist thresholds $t_1, t_2, \ldots, t_s \in \mathbb{R}$ such that for each $I \in 2^{[s]}$ there exists a function $f_I \in \mathcal{F}$ such that $I = \{i \in [s] \mid f_I(x_i) \leqslant t_i \}$.
The \emph{pseudo-dimension} of $\mathcal{F}$, denoted by $\pdim(\mathcal{F})$, is the size $|S|$ of a largest set $S \subseteq X$ that is shattered by $\mathcal{F}$.
\end{definition}

\begin{theorem}[cf.~\cite{Anthony02}, {\cite[Theorem~2.1]{Morgenstern15}}]
\label{thm:anthonybarlett}
Let $\mathcal{F} \subseteq [0, U]^X$ be a set of bounded real-valued functions from a domain $X$, and let $\mathcal{D}$ be a distribution over $X$. Let $x_1, x_2, \ldots, x_k \in X$ be a set of $k = \Theta\big((\nicefrac{U}{\epsilon})^2 (\pdim(\mathcal{F}) \log(\nicefrac{U}{\epsilon})+ \log(\nicefrac{1}{p}))\big)$ independent samples from $\mathcal{D}$. Then, with probability at least $1-p$, for every $f \in \mathcal{F}$, the average of $f$ over the samples approximates the expectation of $f$ over $\mathcal{D}$ within an additive term at most $\epsilon$, i.e.,
\[\bigg| \frac{\sum_{i \in [k]} f(x_i)}{k} - \mathbb{E}_{x \sim \mathcal{D}} f(x) \bigg| \leqslant \epsilon.\]
\end{theorem}

Consider the following set of functions.
\begin{multline*}
\mathcal{H} = \big\{h_f : \mathbb{R}_{\geqslant 0}^E \ni c \mapsto ||f - f^*(c)||_1 \in \mathbb{R} \ \big| \ f \in \mathbb{R}_{\geqslant 0}^E \text{ such that} \\
f \text{ satisfies flow preservation and } ||f||_1 \leqslant 2c_{\max}|E|\big\}
\end{multline*}

Note that, for a fixed flow prediction $f$, function $h_f$ maps each possible input of the maximum flow problem (i.e., a list of edge capacities) to the $\ell_1$ error of the prediction on this input. In other words, $\cost_{c_1,\ldots,c_k}(f) =  \frac{1}{k} \sum_{i \in [k]}  h_f(c_i)$ and $\cost_{\mathcal{D}}(f) = \mathbb{E}_{c\sim\mathcal{D}} h_f(c)$.

In order to apply Theorem~\ref{thm:anthonybarlett}, we first need to upper-bound $\pdim(\mathcal{H})$.
It follows immediately from the definition that the pseudo-dimension is monotone, i.e., if $\mathcal{F} \subseteq \mathcal{G}$, then $\pdim(\mathcal{F}) \leqslant \pdim(\mathcal{G})$, see, e.g., \cite[Section~6.8, Exercise~1]{ShalevShwartz14}.
Hence, we upper-bound $\pdim(\mathcal{H})$ by the pseudo-dimension of a superset of $\mathcal{H}$, which is obtained by dropping the flow preservation and $\ell_1$-norm requirements on the prediction vector $f$, i.e.,

\[\mathcal{H}' = \big\{h_f : \mathbb{R}_{\geqslant 0}^E \ni c \mapsto ||f - f^*(c)||_1 \in \mathbb{R} \;\big|\; f \in \mathbb{R}_{\geqslant 0}^E \big\}.\]

The pseudo-dimension of $\mathcal{H}'$ will be in turn upper-bounded by the pseudo-dimension of the following class of functions analyzed by Dinitz et al.\footnote{We note that the bound of Theorem~\ref{thm:Hn} follows also from~\cite[Theorem~2.1]{Bartlett98}, because, for a fixed $y \in \mathbb{R}^n$, the function $f_y(x) = ||y - x||_1$ can be computed by a neural network with one hidden layer, a piecewise linear activation function, and $O(n)$ parameters.}

\begin{theorem}[cf.~{\cite[Theorem~7]{Dinitz21}}]
\label{thm:Hn}
Let $\mathcal{H}_n = \{f_y \mid y \in \mathbb{R}^n\}$, where $f_y : \mathbb{R}^n \to \mathbb{R}$ is defined by $f_y(x) = ||y - x||_1$. The pseudo-dimension of $\mathcal{H}_n$ is at most $O(n \log n)$.
\end{theorem}

Indeed, to see that $\pdim(\mathcal{H}') \leqslant \pdim(\mathcal{H}_{|E|})$, observe that 
if $\{c_1, c_2, \ldots, c_s\}$ is shattered by $\mathcal{H}'$, then $\{f^*(c_1), f^*(c_2), \ldots, f^*(c_s)\}$ is shattered by $\mathcal{H}_{|E|}$, and the bound follows.\footnote{Like monotonicity, this is a general property of the pseudo-dimension, it follows immediately from the definition, and does not use any specific property of function $f^*$, but we were not able to find a suitable reference.}

Second, to use Theorem~\ref{thm:anthonybarlett}, we need a bound on the maximum value of the considered functions. For every flow prediction $f$ satisfying $||f||_1 \leqslant 2c_{\max}|E|$ and for every capacities $c\in\supp(\mathcal{D})$ we have 
\[
h_f(c) = ||f-f^*(c)||_1 \leqslant ||f||_1 + ||f^*(c)||_1 \leqslant 2c_{\max}|E| + c_{\max}|E| = 3c_{\max}|E|.
\]

Now, having upper bounds on $\pdim(\mathcal{H})$ and on values of $h_f$'s, we are ready to prove the lemma.

\begin{proof}[Proof of Lemma~\ref{lem:altsample}]
Let $U=3c_{\max}|E|$. Consider the following class of functions.
\[\bar{\mathcal{H}} = \big\{\bar{h}_f : \mathbb{R}_{\geqslant 0}^E \ni c \mapsto \min\{h_f(c), U\} \in [0,U] \;\big|\; h_f \in \mathcal{H} \big\}\]
Clearly, any set that is shattered by $\bar{\mathcal{H}}$ is also shattered by $\mathcal{H}$, and hence
\[\pdim(\bar{\mathcal{H}}) \leqslant \pdim(\mathcal{H}) \leqslant \pdim(\mathcal{H}') \leqslant \pdim(\mathcal{H}_{|E|}) \leqslant O(|E| \log |E|).\]
We apply Theorem~\ref{thm:anthonybarlett} to $\bar{\mathcal{H}}$, with $\epsilon=1$ and $p=\nicefrac{1}{\poly(c_{\max}|E|)}$, and we conclude that $k=\Theta(c_{\max}^2 |E|^3 \log (c_{\max} |E|) \log E)$ samples suffice to guarantee that with high probability
\[| \cost_{c_1,\ldots,c_k}(f) - \cost_{\mathcal{D}}(f) | = \bigg| \frac{\sum_{i \in [k]} h_f(c_i)}{k} - \mathbb{E}_{c \sim \mathcal{D}} h_f(c) \bigg| \leqslant 1.\]
for all $f$ with $||f||_1 \leqslant 2c_{\max}|E|$.
\end{proof}

\end{document}